\documentclass[brochure,english,12pt]{bourbaki}
\usepackage[T1]{fontenc} 
\usepackage[matrix,arrow]{xy}
\usepackage{amssymb,amsfonts,amsmath,footnote}
\usepackage[francais]{babel}
\usepackage{MnSymbol}
\usepackage{url}
\usepackage{listings}
\date{June 2014}
\bbkannee{66\`eme ann\'ee, 2013-2014}
\bbknumero{1086}
\title{Developments in Formal Proofs}
\author{Thomas C. HALES}
\address{University of Pittsburgh\\
Department of Mathematics\\
Pittsburgh, PA 15260-2341\\
U.S.A.\\}
\email{hales@pitt.edu}

\theoremstyle{plain}

\newtheorem{theorem}[equation]{Theorem}

\def\op#1{{\operatorname{#1}}}
\def\ring#1{{\mathbb{#1}}}

\def\bool{\op{\tt bool}}
\def\tc{\hbox{:}}

\def\sl{\mathfrak{sl}}

\def\T{{\mathcal T}}

\def\O{{\mathcal O}}

\begin{document}

\maketitle

{

\narrower{\it Si la math\'ematique formalis\'ee \'etait
aussi simple que le jeu
d'\'echecs, \ldots il n'y aurait plus qu'\`a r\'ediger nos d\'emonstrations dans ce
langage, comme l'auteur d'un trait\'e d'\'echecs \'ecrit dans sa notation.\ldots
Mais les choses sont loin d'\^etre aussi faciles, et point n'est besoin d'une
longue pratique pour s'apercevoir qu'un tel projet est absolument irr\'ealisable.
 -- Bourbaki, 1966 \cite{bourbaki1966theorie}}. 

}

\bigskip

A proof assistant is interactive computer software that humans use to
prepare scripts of mathematical proofs.  These proof scripts can be
parsed and verified directly from the fundamental rules of logic and
the foundational axioms of mathematics.  The technology underlying
proof assistants and formal proofs has been under development for
decades and grew out of efforts in the early twentieth century to
place mathematics on solid foundations.  Proof assistants have been
built upon various mathematical foundations, including
Zermelo-Fraenkel set theory (Mizar), Higher Order Logic (HOL), and
dependent type theory (Coq) \cite{Mizar}, \cite{HOLL}, \cite{Coq}.  A
{\it formal proof} is one that has been verified from first principles
(generally by computer).

This report will focus on three particular technological advances.
The HOL Light proof assistant will be used to illustrate the design of
a highly reliable system.  Today, proof assistants can verify large
bodies of advanced mathematics; and  as an example, we will turn to the
formal proof in Coq of the Feit-Thompson Odd Order theorem in group
theory.  Finally, we will discuss advances in the automation of
formal proofs, as implemented in proof assistants such as Mizar, Coq,
Isabelle, and HOL Light.

\section{Building a trustworthy system with HOL Light}

HOL Light is a lightweight implementation of a foundational system
based on Higher Order Logic (HOL).  Because it is such a lightweight
system, it is a natural system to use for explorations of the
reliability of formal proof assistants.

\subsection{Naive type theory}

HOL, the foundational system of mathematics that we describe in this
section, is based on a simply typed $\lambda$-calculus.  This
subsection describes a simple type theory in naive terms.

A salient feature of set theory is that it is so amorphous; everything
is a set: ordered pairs are sets, elements of sets are sets, and
functions between sets are sets.  Thus, it is meaningful in set theory
to ask bizarre questions such as whether a Turing machine is a minimal
surface.  In type theory, the very syntax of the language prohibits
this question.  Computer systems benefit from the extra structure
provided by types.

Naively, a simple type system is a countable collection of disjoint
nonempty sets called types.  The collection of types satisfies a
closure property: for every two types $A$ and $B$, there is a further
type, denoted $A\to B$, that can be identified with the set of
functions from $A$ to $B$.

In addition to types, there are terms, which are thought of as
elements of types.  Each term $t$ has a unique type $A$.  This
relationship between a term and its type is denoted $t:A$.  In
particular, $f:A\to B$ denotes a term $f$ of type $A\to B$.

There are variables that range over types called {\it type variables},
and another collection of variables that run over terms.

\subsection{Models of HOL}

The naive interpretation of types as sets can be made precise.  We
build a model of HOL in Zermelo-Fraenkel-Choice (ZFC) set theory to
prove that HOL is consistent assuming that ZFC is.  In this section,
we review this routine exercise in model theory.  At the same time, we
will give some indications of the structure of HOL Light.  See
\cite{harrison2009hol} for a more comprehensive introduction to HOL
Light.

The interpretation of variable-free types as sets is recursively
defined.  We use a superscript $M$ to mark the interpretation of a
type as a set.  Specifically, the types in HOL are generated by the
boolean type $\bool$ (which we interpret as a set
$\bool^M=\{\downvdash ,\upvdash \}$ of cardinality two with labeled
elements representing true and false) and the infinite type $I$
(which we interpret as a countably infinite set $I^M$).  Recursively,
for any two variable-free types $A$ and $B$, the type $A\to B$ is
interpreted as the set $(A\to B)^M$ of all functions from $A^M$ to
$B^M$.  We can arrange that the sets interpreting these types are all
disjoint.

In summary so far, we fix an interpretation $M$, determining a
countable collection $\T =\{A^M\}$ of nonempty sets in ZFC.  We now
extend our interpretation $M$ to a {\it valuation} $v=(M,v_1,v_2)$,
where $v_1$ is a function from the set of type variables in HOL to
$\T$, and $v_2$ is a function from the set of term variables in HOL to
$\bigcup\T$.  The valuation $v$ extends recursively to give a mapping
that assigns a set $A^{v}\in \T$ to every type $A$.  We require $v_2$
to be chosen so that whenever $x$ is a variable of type $A$, then
$x^{v_2} \in A^v$.  The valuation $v$ extends recursively to give a
mapping on all terms:
\[
t\mapsto t^v \in A^v \in \T, \quad \text{for all } t:A.
\]

For example, for every type $A$, there is a HOL term $(=)$ of type
$A\to (A\to \bool)$ representing equality for that type.\footnote{The
  convention in HOL is to curry functions: using the bijection $X^{Y
    \times Z} = (X^Z)^Y$ to write a function whose domain is a product
  as a function of a single argument taking values in a function
  space. In particular, equality is a curried function of type $A\to
  (A \to\bool)$ rather than a relation on $A\times A$.}  This term is
interpreted as the function in $(A\to (A\to\bool))^v$ that maps $a\in
A^{v}$ to the delta function $\delta_a$ supported at $a$ (where the
support of the function means the preimage of $\downvdash$).

A {\it sequent} is a pair $(L,t)$, traditionally written $L\vdash t$,
where $L$ is a finite set of terms called the {\it assumptions}, and
$t$ is a term called the {\it conclusion}.  The terms of $L$ and $t$
must all have type $\bool$.  If $L$ is empty, it is omitted from the
notation.

If $L$ is a finite set of boolean terms, and if $v$ is a valuation
extending $M$, write $L^{v}$ for the corresponding set of elements of
the set $\bool^M$.  We say a sequent $L \vdash t$ is {\it logically
  valid} if for every valuation $v$ for which every element of $L^{v}$
is $\downvdash \in \bool^M$, we also have $t^{v}=\downvdash $ in
$\bool^M$.

A {\it theorem} in HOL is a sequent that is generated from the
mathematical axioms and rules of logic.  There is a constant ${\tt
  FALSE}$ in HOL.  The following amounts to saying that HOL does not
prove {\tt FALSE}.

\begin{theorem} If ZFC is consistent, then HOL is consistent.
\end{theorem}

\begin{proof}[Proof sketch]
  We give the proof in ZFC.  Here, HOL is treated purely syntactically
  as a set of strings in a formal language.

  We run through the rules of logic of HOL one by one and check that
  each one preserves validity.\footnote{There are ten such rules,
    giving the behavior of equality, $\lambda$-abstractions,
    $\beta$-reduction, and the discharge of assumptions.  For
    reference purposes, an appendix lists the inference rules of HOL.
    The analysis in this section omits the rules for the creation of
    new term constants and types.} For example, the reflexive law of
  equality in HOL states that for any term $t$ of any type $A$, we
  have a theorem $\vdash t = t$.  By the interpretation of equality
  described above, under any valuation $v$, this equation is
  interpreted as the value $\delta_{t^{v}}(t^{v})\in\bool^M$, which is
  $\downvdash $.  Hence the reflexive law preserves validity.  The
  other rules (transitivity of equality, and so forth) are checked
  similarly.

  We may well-order each set in the collection $\T$.  HOL posits a
  choice operator of type $(A\to\bool)\to A$ for every type $A$.  The
  well-ordering allows us to interpret HOL's choice operator as an
  operator that maps a function $f\in (A\to\bool)^v$ with nonempty
  support to the minimal element of its support.

  We run through the mathematical axioms of HOL and check their
  validity.  There are only three.  The {\it axiom of infinity}
  positing the existence of an infinite type $I$ is logically valid by
  our requirement to interpret $I$ as a countably infinite set.  The
  {\it axiom of choice} is logically valid by the well-ordering we
  have placed on the sets $A^M$.  The {\it axiom of extensionality}
  also holds in this model, because it holds for sets.  Thus, every
  axiom is logically valid.  Since all axioms are logically valid and
  every rule of inference preserves validity, every theorem is
  logically valid.

  There is a boolean constant $\op{\tt FALSE}$ in HOL. It is defined
  to be the term $\forall p.~p$.  An easy calculation based on its definition
  gives that $\op{\tt FALSE}^v = \upvdash$ for every valuation $v$.
  Hence, $\vdash \op{\tt FALSE}$ is not a logically valid sequent, and
  not a theorem.  This proves HOL consistent.
\end{proof}

\subsection{Computer implementation}

The bare consistency proof is just the beginning.  We can push matters
much further when the logic is implemented in computer code.

The HOL Light system is implemented in the Objective Caml programming
language, which is one of the many dialects of the ML language.  The
language ML (an acronym for Meta Language) was originally designed as
a metalanguage to automate mathematical proof commands \cite{Gor}.  It
is significant that the development of the language and the
development of proof assistants have progressed hand in hand, with
many of the same researchers participating in language design and
formal proofs.  The result is a programming language that can stand up
to intense mathematical scrutiny.

This parallel development of ML and proof assistants also means that
there are striking similarities between the syntax of ML and the
syntax of HOL.\footnote{HOL is a descendant of the LCF
  theorem prover that spurred the development of ML.}  The code
listing shows a few parallels between OCaml syntax and HOL syntax.

\begin{lstlisting}[keepspaces=true,stringstyle=\tt,basicstyle=\small,frame=single,framesep=8pt,mathescape,morekeywords={type,let,and,in,OCaml,HOL},columns=flexible]
OCaml                          versus    HOL
-----                                    ---
3:int                                    3:num
[0;1;2;3]                                [0;1;2;3]  
let x = 3 and y = 4 in x + y             let x = 3 and y = 4 in x + y
map (fun x -> x + 1) [0;1;2]             map (\x. x + 1) [0;1;2]
\end{lstlisting}



In both ML and HOL every term has a type, and functions $f:A\to B$ of
the right type are required to convert from a term of one type $A$ to
another $B$.  One of Milner's key ideas in the design of ML was to
have an abstract datatype representing theorems.  The strict type
system of the language ML prevents the construction of any theorems
except through a carefully secured kernel that expresses the axioms
and rules of inference.  In a formal proof in HOL, absolutely every theorem -- no
matter how long or how complex -- is checked exhaustively by the
kernel.

\subsection{Verification of the code that implements HOL Light}

The code in the kernel that expresses the rules of HOL is of critical
importance.  Even a minor bug in the kernel might be exploited to
create an inconsistent system.  Fortunately, there are good reasons to
believe that the kernel does not have a single bug.

1.  The kernel is remarkably small.  It takes only about 400 lines of
computer code to express all of the kernel functions, including the
type system, the term constructors, sequents, the rules of inference,
the axioms, and theorems.  For example, it only takes seven lines of
computer code to describe the datatypes for HOL types, terms, and
theorems, as shown in the following listing of code~\cite{HOLL}.

\begin{lstlisting}[keepspaces=true,stringstyle=\tt,basicstyle=\small,frame=single,framesep=8pt,mathescape,morekeywords={type},columns=flexible]
  type hol_type = Tyvar of string
                | Tyapp of string *  hol_type list

  type term = Var of string * hol_type
            | Const of string * hol_type
            | Comb of term * term
            | Abs of term * term

  type thm = Sequent of (term list * term)
\end{lstlisting}

2.  The code has been written in a clean, readable style and has been
scrutinized by many computer scientists, logicians, and mathematicians
(including me).

3.  The correctness of the kernel has been formally verified, using
the HOL Light proof assistant itself (extended by a large cardinal)
\cite{HaSelf}.\footnote{A large cardinal axiom gives the existence of a large
  type corresponding to the set $\bigcup\T$ that we used above in the
  construction of a model.  By G\"odel, we do not expect to construct
  a model of HOL in HOL except by adding an axiom to strengthen the
  system.}  Specifically, a model of HOL can be built inside HOL
itself along the same lines as the model of HOL in ZFC described
above.

This formal verification of HOL in HOL goes further than the
construction of a model.  It also checks that the code implementing
the logic is bug free.  The code verification is based on the
parallels mentioned above between the metalanguage and HOL itself,
allowing the OCaml source code for the HOL kernel to be translated
back into HOL for verification.  A stricter standard of code
verification, based on the semantics of the programming language, is
discussed in the next subsection.  The proof of HOL in HOL removes
most practical doubts about the correctness of the kernel.

As independent corroboration, the correctness proof of the kernel of
HOL Light has been automatically translated into the HOL Zero and HOL4
assistants and reverified there \cite{adams2010introducing}.

\subsection{HOL in machine code}

The formal verification of HOL in HOL does not settle the issue of
trust once and for all.  The reliability of the proof assistant
ultimately rests on the entire computer environment in which the
software operates, including the semantics of programming languages,
compilers, operating systems, and hardware.  These issues should be a
concern of every mathematician who cares about the foundations of
mathematics at a time when the practice of mathematics is gradually
migrating to computers.

The current working goal of researchers is to create an unbroken
formally-verified chain extended from the HOL Light kernel all the way
down to machine code.  Most of the links in the chain have been forged.\footnote{A
  brilliant success has been the construction of a formally verified C
  compiler~\cite{CC}. Another remarkable project is the full formal
  verification of an operating system kernel~\cite{sel4}.  Concerning
  the formal verification of Coq and Milawa (an ACL2-like system), see
  \cite{barras2010sets} \cite{myreen2012reflective}.  In this survey
  article, we focus on the work done on formal verification related to
  the ML programming language, because it fits more closely with our
  narrative of building a trustworthy system in HOL.}

CakeML is a dialect of ML with mathematically rigorous operational
semantics.  According to its designers, ``Our overall goal for CakeML
is to provide the most secure system possible for running verified
software and other programs that require a high-assurance platform''
\cite{CakeML}.  The CakeML team has improved the HOL in HOL
verification \cite{myreen2013steps}, \cite{myreen2013proof},
\cite{kumar2014hol}.  This work closes various gaps in the
verification of the OCaml implementation of HOL, such as {\it object
  magic} (a mechanism that defeats the OCaml type system) and {\it
  mutable strings} (which allow a theorem to be edited to state
something different from what was proved).  The HOL system can be
extended by adding new definitions and new types.  The formal
verification now covers these extensions.  It verifies the soundness
of an implementation of the HOL kernel in CakeML, according to the
formally specified operational semantics of CakeML.

In related work, a verified compiler has been constructed for the
CakeML language \cite{CakeML}, \cite{sarkar2009semantics}.  This
brings us close to end-to-end verification of HOL Light, from its
high-level logical description down to execution in machine code.

\subsection{Disclaimers}

At the conclusion of this section, we make the usual disclaimers.
Without exception, all physical devices are prone to failure.  Soft
errors (typically caused by alpha particle interactions between memory
and its environment) produce a steady stream of errors depending on
complex factors such as hardware architecture and the height above sea
level at which the calculation is performed.  A formal proof in HOL of
the correctness of HOL carries the evident dangers of
self-justification.  Proofs of correctness are made relative to
mathematically precise idealized descriptions of things rather than
the physical objects themselves.

Notwithstanding all these issues, formalization reduces defect rates
in proofs to levels that are simply not possible by any other
available process.  Now that the formal proof of HOL in HOL has been
translated into other systems, the dangers of self-justification are
minimal.  Overall, formalized mathematics can now claim to be
orders of magnitude more reliable than traditionally refereed papers.

\section{Advanced developments in Coq: The Odd Order theorem}

In the previous section, we discussed the construction of a reliable
proof assistant.  In this section, we turn to a different proof
assistant, Coq, and look at the formalization of group theory.  We
warn the reader that the Coq system, which is based on the Calculus of
Inductive Constructions, is significantly different from the HOL
system in the previous section~\cite{Coq2014}, \cite{CiC}.

\bigskip

Feit and Thompson published their famous theorem in 1963~\cite{FT}.

\begin{theorem}[The Odd Order Theorem, Feit-Thompson]  All finite groups of odd order are solvable.
\end{theorem}

\begin{lstlisting}[keepspaces=true,stringstyle=\tt,basicstyle=\small,frame=single,framesep=8pt,mathescape,morekeywords={Theorem},columns=flexible]
Theorem Feit_Thompson (gT : finGroup Type) (G : {group gT}) :
  odd #|G| -> solvable G.
\end{lstlisting}

Solomon writes this about the significance of the Odd Order theorem,
``This short sentence and its long proof were a moment in the
evolution of finite group theory analogous to the emergence of fish
onto dry land.  Nothing like it had happened before; nothing quite
like it has happened since'' \cite{Sol01}.

The Odd Order paper broke through various barriers that cleared the
way for a remarkably fruitful massive research collaboration that
eventually led to the classification of finite simple groups.
Significantly, the 255 page Odd Order proof triggered an avalanche of
long complex proofs related to the classification, culminating in the
1221 page classification of quasi-thin groups by Aschbacher and
Smith~\cite{aschbacher2004classification}.

Background material for the proof appears in textbooks {\it Finite
  Groups}~\cite{gorenstein2007finite}, {\it Finite Group Theory}
\cite{aschbacher2000finite}, and {\it Character Theory of Finite
  Groups}~\cite{isaacs2013character}.  The necessary background
includes a basic graduate-level understanding of rings, modules,
linear and multilinear algebra (including direct sums, tensor products
and determinants); fields, algebraic closures, and basic Galois
theory; the structure theorems of Sylow and Hall; Jordan-H\"older;
Wedderburn's structure theorem for semisimple algebras; representation
theory with induced representations, Schur's lemma, Clifford theory,
and Maschke's theorem; and character theory including Frobenius
reciprocity, and orthogonality.

I will not say much about the actual proof of the Odd Order theorem.
We have now had more than fifty years to assimilate the ideas of the
proof.  There are numerous surveys of the proof
\cite[p. 450]{gorenstein2007finite}, \cite{glauberman1999new},
\cite{thompson1968nonsolvable}, \cite{Sol01}.  The original proof of
Feit and Thompson was later reworked and simplified in two books
\cite{bender1994local}, \cite{peterfalvi2000character}.

In very brief terms, the proof starts by assuming a minimal
counterexample to the statement.  This counterexample will be a finite
simple group $G$ of odd order in which every proper subgroup is
solvable.  Each maximal subgroup of $G$ is {\it $p$-local}; that is,
the normalizer of a nontrivial subgroup of $G$ of $p$-power order, for
some prime $p$.  The first major part of the proof of the Odd Order
theorem consists in establishing restrictions on the structure of the
maximal subgroups and their embeddings into $G$.  In the special case
when $G$ is a {\it $CN$-group} (a group whose centralizers of
non-identity elements are all nilpotent), the maximal subgroups are
{\it Frobenius groups} \cite{FHT}.  A Frobenius group is a nontrivial
semidirect product $K\rtimes H$, where $H$ is disjoint from its
conjugates, and $H$ is its own normalizer
\cite[Th. 7.7]{gorenstein2007finite}.  In the general case, the
strategy is to prove that as many of the maximal subgroups as possible
are as close as possible to being Frobenius groups.  This strategy
encounters many exceptions and detours, but eventually the local
analysis shows that the maximal subgroups are mostly
Frobenius-like \cite[Sec. 16]{bender1994local}.


The second major portion of the proof uses the complex character
theory of $G$ to obtain inequalities over the real numbers that
restrict and ultimately eliminate all possibilities for $G$.

The final part of the proof uses generators and relations to remove a
special case in which there are maximal subgroups isomorphic to the
group of all permutations of a finite field $\ring{F}$ of the form $x
\mapsto a\, \sigma (x) + b$, where $\sigma$ is a field automorphism,
$a,b\in\ring{F}$, and $a$ is an element of norm $1$. It is in this
part of the proof that Galois theory is most relevant.

\subsection{Formal verification}

The Odd Order theorem has been formally verified in the Coq proof
assistant by a team led by Gonthier  \cite{gonthier2013machine}.
This is an extraordinary milestone in the history of formal proofs.
What is particularly significant about the formalization itself?

1. The Odd Order theorem itself has never been seriously questioned,
but the premature announcement of the classification of finite simple
groups drew sustained criticism.  Gorenstein wrote, ``In February
1981, the classification of the finite simple groups was completed''
\cite[page 1]{gorenstein2007finite}; and yet essential work on the
classification continued for decades after that date
\cite{aschbacher2004status}.  Serre wrote in 2004 that ``for years I
have been arguing with group theorists who claimed that the
`Classification Theorem' was a `theorem', i.e.\ had been proved''
\cite{raussen2004interview}.  Because of the prominence of the Odd
Order theorem within the classification, this formal proof plants a
flag in the middle of the larger classification project.


2. Traditional methods of refereeing mathematical research become
strained, when proofs are unusually long or computer assisted.  There
is a Wikipedia page listing numerous proofs in mathematics that set
records for length~\cite{WikiLong}.  Long papers cluster in certain
areas such as finite group theory related to the classification, the
Langlands program, and graph theory.  The Odd Order theorem, the
Four-Color theorem (Gonthier's previous formalization project), and
the Kepler conjecture are all on that list.  These three recent
formalization projects send a clear defiant message to mathematicians:
no matter how long or how complex your mathematical proofs may be, we
can formalize them.

Absolutely no technological barriers prevent the formalization of
large parts of the mathematical corpus.  The issues now are how to
make the technology more efficient, cost effective, and user friendly.

\subsection{Constructive proof of the Odd Order theorem}

The formal proof of the Odd Order theorem is based on the
second-generation proof described in two books cited above
\cite{bender1994local}, \cite{peterfalvi2000character}.  If we were to
translate the formal proof of Odd Order back into a humanly readable
book, the most notable difference would be that the original proof
uses classical logic, but the formal proof uses constructive logic.
In particular, a constructive proof avoids proofs by contradiction and
does not invoke  the law of excluded middle $\phi\lor \neg \phi$ as
a general principle.

Several different strategies were used to translate the paper proof
into a constructive computer proof.

1.  Often, mathematicians use proof by contradiction out of sheer
laziness when a direct proof would work equally well.  ``Let $G$ be
the finite group of minimal order that is a counterexample'' is
replaced with an induction on the order of the group, and so forth.

2. In \cite{peterfalvi2000character}, the character theory of finite
groups relies heavily on vector spaces over $\ring{C}$ and complex
conjugation.  In the constructive formal proof, the corresponding
vector spaces over an algebraic closure $\bar{\ring{Q}}$ are used.  An
algebraic closure of $\ring{Q}$ is obtained as the union of an
increasing tower of number fields $\ring{Q}(\alpha_i)$, for
$i\in\ring{N}$, each an explicitly constructed splitting field over
the preceding one.  Complex conjugation is replaced with conjugation
with respect to a maximal real subfield of $\bar{\ring{Q}}$, also
constructed as an explicit union of real subfields of finite degree
over $\ring{Q}$.

3.  Some constructions are justified by the decision procedure for the
first-order theory of algebraically closed
fields~\cite{cohen2010formal}, \cite{Ha09}.  For example, the socle of
a module is defined as the sum of its simple submodules.  This
decision procedure gives a constructive test for the simplicity of a
submodule, leading to a constructive definition of the socle (for a
finitely generated $k[G]$-module where $G$ is a finite group and $k$
is an algebraically closed field) \cite{gonthier2011point}.

On a related note, the construction of a simple submodule of a given
module $M$ requires choice.  By confining attention to modules that
are countable as sets, countable choice suffices, which is provable in
Coq from the well-foundedness of the natural numbers.

4.  Various intermediate results in Odd Order theorem use classical
logic, but the use of classical logic is eliminated from the final
statement of the Odd Order theorem.  The wedding of constructive and
classical logics is arranged through a predicate $\op{\it
  classically}~\phi$ that marks every result proved by classical
logic.  The definition of $\op{\it classically}~\phi$ is
\[
\op{classically}~\phi : Prop := \forall {b:\bool}, (\phi\to b) \to b.
\]
When $\phi$ is quantifier free, this is essentially the usual double
negation translation of a classical formula into constructive logic.
The definition of {\it classically} works in such a way that $\phi$
may be used as an assumption in the proof of any boolean goal $b$,
whenever $\op{\it classically}~\phi$ is known.

5.  Function extensionality does not hold in Coq; that is, it is not
provable that two functions are equal, if they are provably equal on
every input.  This creates various complications.  There are no
general quotient types: types corresponding to the set of classes on a
type under an equivalence relation.  Instead, Coq introduces {\it
  setoids} -- a set together with an equivalence relation on that set
-- without passing to the quotient.  Unfortunately, setoids carry a
certain overhead, which was not acceptable in the Odd Order proof,
which makes ubiquitous use of quotient groups.

The group theory is developed in the context of {\it finite types}
with {\it decidable equality}; (that is, the type is equipped with
boolean procedure that decides whether any two elements of that type
are equal).  In this context, function extensionality holds and
genuine quotients groups can be formed.

6.  For all its advantages, at times the type theory of Coq is best
forgotten.  In particular, in arguments involving a finite number of
finite groups, there would be excessive overhead in a separate type
for each group, a separate binary operation on each group, and
explicit homomorphisms embedding subgroups into groups.  For such
arguments, the finitely many finite groups are often considered as
subgroups of a larger ambient finite group with a common binary
operation and type.  These ambient group arguments are one of the
biggest stylistic differences between the Coq proof and the original.

\subsection{Library of abstract algebra}

Most of the work for this project went into the development of the
libraries in abstract algebra and the related computer infrastructure.
The libraries include formal proofs of all the necessary background
material described at the beginning of this section, from Frobenius
reciprocity to Wedderburn's theorem.  From a software engineering
point of view, it has been a major undertaking to get the computer to
understand algebra at a level comparable to that of a working
mathematician, and for it all to be formally justified.  This includes
much of the implicit domain knowledge that is required to read a
proof.  When we write $g*h$, we expect the computer to infer the
correct binary operation.  When we take the intersection $K\cap H$ of
two subgroups, we expect it to go without saying that the intersection
is again a subgroup.

\section{Automating formalization}

One way to make proof assistants more usable is to increase the amount
of automation.  Ideally, the human should provide the high level
structure of a proof as it is done in a traditionally published paper,
and the computer should use search algorithms to insert the low-level
reasoning.  As technology has developed, computers have become capable
of providing more and more of the low-level reasoning.  For an
overview of automated reasoning in formal proofs, see \cite{Ha09}.

As we explained in the first section of this report, it is a matter of
great importance for the underlying logic of the proof assistant to
be sound and for the implementation in code to be free of bugs.  Proof
search algorithms are based on fundamentally different design
considerations than the proof assistants themselves.  With search
algorithms, we can take a more relaxed approach, with the
understanding that their results must eventually be verified by the
proof assistant before being admitted as part of formal proof.  

There is no general decision procedure for mathematics as a whole.
Some decision procedures such as the Presburger algorithm for the
additive first-order theory of natural numbers or quantifier
elimination for real closed fields only have a limited practical value
for formal proofs because they are so slow.

There are dangers in combining partial automation with human
interaction.  There is a psychological tendency for the human to
``wait and see'' rather than race against the computer, when the human
and computer are both engaged in the same task.  An automated
procedure might make sudden complicated changes to the proof state
that strand the user in uncharted territory.  For these reasons, it is
good to have procedures that fail quickly when they fail and that take
a fully-documented circumscribed step forward when they succeed.

\subsection{First-order logic}

A {\it first-order formula with equality in the predicate calculus} is
a formula built from variables $x,y,\ldots$, logical operations $\neg,
\land, \Rightarrow,\ldots$, $n$-ary function symbols $f,g\ldots$,
$n$-ary predicate symbols $P, Q,\ldots$ and quantified variables
$\forall x, \exists y,\ldots$, according to syntactic rules that we
will not spell out here.  For example,
\[
((\exists x.~P(x)) \Rightarrow (\forall y.~Q(y))) \Leftrightarrow (\forall x\forall y.~P(x)\Rightarrow Q(y))
\]
is a first-order formula.  In contrast to higher-order logics
discussed earlier, in a first-order formula, quantifiers over function
symbols and predicate symbols are not allowed.

We are interested in algorithms that are {\it refutation complete};
that is, given the input of an unsatisfiable first-order formula, the
algorithm outputs a proof of its unsatisfiability.  A refutation
complete algorithm can also produce proofs of logically valid
formulas, by refuting the negation of the formula.

What has been understood for a long time are rules of inference ({\it
  resolution} together with special inferences that deal with
equality) that are refutation complete.  Recent research seeks to find
efficient ways to manage and contain the vast explosion of clauses
that can result by repeated application of the inference rules.
Poorly designed algorithms exhaust available memory before a
refutation is found.

Today, automated first-order theorem provers can handle problems with
thousands of axioms~\cite{bachmair2001resolution}, \cite{Ha09}.  In
recent years, the {\it Vampire} theorem prover has dominated the
annual competition, with other strong contenders such as the {\it
  SPASS} and {\it E} theorem
provers~\cite{kovacs2013first},~\cite{riazanov2002design},
\cite{WeidenbachDFKSW09}, \cite{Sch02-AICOMM}.



There are some automated theorem provers based on higher-order
logic 
but the most prevalent practice is to translate problems from
higher-order logic into first-order logic, solve them there, then
translate the answers back into higher-order logic, with an automated
reconstruction of a formal proof inside the proof
assistant.\footnote{There are many issues that come up in translating
  formulas in higher-order logic into first-order logic that I will
  not discuss here.  In particular, some of the procedures might not
  be logically sound.}

An early example of translating proofs in this way is Harrison's MESON
procedure (based on model elimination), implemented in HOL Light.  In
this approach, the human supplies all the relevant theorems, and the
automated procedure generates the logical glue that combines the given
theorems into a proof.

Paulson advocates shipping the goal together with large libraries of
theorems to a first-order theorem prover, and letting it figure out which
of many theorems to use in the proof~\cite{Paar}.  The human is no longer
forced to search through the libraries to find the relevant theorems.
The method is called a {\it sledgehammer} for its ability to deliver a
powerful blow and its complete lack of finesse.

In a refinement of this approach, an automated heuristic procedure
selects a few hundred theorems deemed to be the most relevant and
ships those to the theorem prover.  For example, to prove a new trig
identity, we might select other trig identities as likely to be
relevant.  There is an art to selecting relevant premises wisely, and
machine learning algorithms can be trained to do this effectively
\cite{urban2010evaluation},  
\cite{kaliszyk2012learning}, \cite{kaliszyk2013automated},
\cite{alama2014premise}.

Here is one simple example of a theorem proved in this
way~\cite{kaliszyk2014learning}.  The theorem states that the convex
hull of any three points in $\ring{R}^3$ is a set of measure zero.
The automated procedure was able to comb through large libraries of
previously established theorems to locate the following relevant
facts: the convex hull is a subset of the affine hull; the affine hull
of three points in $\ring{R}^3$ is a set of measure zero; and a subset
of a set of measure zero has measure zero.  The procedure then
combines these facts with the necessary logic to give a fully
automated formal proof.

Sledgehammers and machine learning algorithms have led to visible
success.  Fully automated procedures can prove 40\% of the theorems in
the Mizar math library, 47\% of the HOL Light/Flyspeck libraries, with
comparable rates in Isabelle \cite{DBLP:journals/corr/KaliszykU13b},
\cite{kaliszyk2014learning}, \cite{bohme2010sledgehammer}.  These
automation rates represent an enormous savings in human labor.

\subsection{SAT Solvers and SMT}

Boolean satisfiability (SAT) solvers implement algorithms that test
for the satisfiability of formulas in propositional
logic~\cite{marques2009conflict}.  Is there an assignment of truth
values to propositional variables for which a given formula evaluates
to true?  SMT ({\it satisfiability modulo theories}) algorithms combine
the propositional reasoning of SAT solvers with reasoning within a
given theory~\cite{barrett2009satisfiability}.  The input to the
algorithm is a propositional formula in which the boolean variables
have been replaced with predicates from the given first-order theory.  An SMT
algorithm searches for a valuation of the predicates in the theory
that satisfies the propositional formula.

For example, the following calculation comes up in the proof of the
Odd Order theorem.  Let $G$ be a finite group, and let $(\beta_{ij})$
be a matrix of virtual characters, with each entry a linear
combination $\pm \chi \pm \chi' \pm \chi''$ of three distinct irreducible
characters.  Assume the matrix has at least four rows and two columns.
Assume the inner product relations
\[
\langle \beta_{ij},\beta_{i'j'}\rangle = \delta_{ii'} + \delta_{jj'},\quad \text{for } (i,j)\ne (i',j'),
\]
with respect to the usual inner product on class functions with an orthonormal basis consisting
of irreducible characters.  
The conclusion is that  the virtual characters in each column of the matrix have a common irreducible character as constituent
with a common sign.

A moment's reflection reveals that it is enough to show the conclusion
for arbitrary $4\times 2$ submatrices.  On each $4\times 2$ block, up
to renaming the irreducible characters, a finite enumeration gives all
ways to express the entries $\beta_{ij}$ as a signed combination of
three irreducibles.  Thus, the proof reduces to a case analysis.
Symmetry arguments further reduce the number of cases.  The case
analysis is done by hand in~\cite{peterfalvi2000character}.  In the
formal proof, Th\'ery programmed the claim into a quantifier-free
problem with uninterpreted symbols in an SMT solver.  The SMT solution
was then transcribed into a formal proof in
Coq~\cite{gonthier2013machine}.  We list a code snippet from the
formal proof of the Odd Order theorem.  
The first line of code asserts unsatisfiability:  no irreducible character
$\chi_1$ appears with positive coefficient in $\beta_{ii}$ and
negative coefficient in $\beta_{ji}$.

\begin{lstlisting}[keepspaces=true,stringstyle=\tt,basicstyle=\small,frame=single,framesep=8pt,mathescape,morekeywords={Let},columns=flexible]
Let unsat_J : unsat $\models$ & x1 in b11 & -x1 in b21.
Let unsat_II: unsat $\models$ & x1, x2 in b11 & x1, x2 in b21.
\end{lstlisting}



\subsection{Other forms of automation}

Automation in proof assistants takes many forms, including the
evaluation of arithmetic expressions, the verification of polynomial and
vector space identities, and decision procedures for linear real
arithmetic.  

Gr\"obner basis algorithms provide a particularly useful tool to prove
general ring identities. If $S$ is any system of equalities and
inequalities of polynomials that holds in every integral domain of
characteristic zero, then it can be transformed into a Gr\"obner basis
problem over $\ring{Q}$.  Buchberger's algorithm has been
implemented in many of the major proof assistants.  For example, one
line of code suffices to generate a formal proof of the following isogeny of
elliptic curves over $\ring{R}$.  The symbol \verb!&! denotes the function
embedding $\ring{N}$ into $\ring{R}$.

\begin{lstlisting}[keepspaces=true,stringstyle=\tt,basicstyle=\small,frame=single,framesep=8pt,mathescape,morekeywords={None},columns=flexible]
 $\vdash$   $a'$ = &2 * $a$ $\land$  $b'$ = $a$ * $a$ - &4 * $b$ $\land$ 
     $x_2$ * $y_1$ = $x_1$ $\land$  $y_2$ * $y_1^2$ = &1 - $b$ * $x_1^4$ $\land$
     $y_1^2$ = &1 + $a$ * $x_1^2$ + $b$ * $x_1^4$ 
     $\Longrightarrow$ $y_2^2$ = &1 - $a'$ * $x_2^2$ + $b'$ * $x_2^4$
\end{lstlisting}

Other powerful forms of automation in
proof assistants include the certification of linear programs and
nonlinear inequalities over the real numbers~\cite{Solovyev-thesis}.

\section{Final Remarks}

The aim of this report has been to describe some of the recent
developments in formal proofs.  Space and time do not permit a
comprehensive survey, but in this final section, I briefly mention a
few other projects.

\subsection{Homotopy Type Theory}

My report will be directly followed by a report by Coquand on
dependent types and the univalence axiom, so I will be brief in my
remarks on homotopy type theory.

Homotopy type theory (HoTT) is a foundational system for mathematics
that includes dependent type theory, the univalence axiom, and higher
inductive types.  Introductions to homotopy type theory can be found
at \cite{hottbook}, \cite{pelayo2012homotopy}.  For models of HoTT,
see \cite{awodey2007homotopy}, \cite{kapulkin2012simplicial}.  HoTT
has set quotients and function extensionality, giving remedies to some
of Coq's nuisances.

It goes without saying that as mathematicians, we construct the ground
on which we stand; the foundations of mathematics are of our choosing,
subject to only mild constraints such as plausible consistency,
expressive power, and a community of users.  In particular, nothing
but our own limited imaginations prevents us from relocating the
foundations much closer to home.

By being a foundational system that is close to the actual practice of
homotopy theory, HoTT makes the formalization of this branch of
mathematics surprisingly refreshing.  In the last two years many new
formal proofs and constructions have been obtained: loop spaces,
computations of various fundamental groups of spheres, the Freudenthal
suspension theorem, the Seifert-van Kampen theorem, construction of
Eilenberg-Mac Lane spaces~\cite{licataeilenberg}, and the
Blakers-Massey theorem. Formalization of these results in other
systems would have been much more labor intensive.  A new line of
research develops homotopy theory within HoTT foundations.

As an example, we list Grayson's code that constructs the classifying
space as the type of a torsor in HoTT~\cite{Ktheory}.  We include his
proof that the classifying space $BG$ is connected.  I challenge any
other system to pass from the foundations of math to classifying
spaces so directly and elegantly!

\begin{lstlisting}[keepspaces=true,stringstyle=\tt,basicstyle=\small,frame=single,framesep=8pt,mathescape,morekeywords={Definition,Lemma,Proof,Defined},columns=flexible]
Definition ClassifyingSpace G := pointedType (Torsor G) (trivialTorsor G).
Definition E := PointedTorsor.
Definition B := ClassifyingSpace.
Definition $\pi$ {G:gr} := underlyingTorsor : E G -> B G.

Lemma isconnBG (G:gr) : isconnected (B G).
Proof. intros. apply (base_connected (trivialTorsor _)).
  intros X. apply (squash_to_prop (torsor_nonempty X)). { apply propproperty. }
  intros x. apply hinhpr. exact (torsor_eqweq_to_path (triviality_isomorphism X x)). 
Defined.
\end{lstlisting}

\subsection{Bourbaki on formalization}

Over the past generation, the mantle for Bourbaki-style mathematics
has passed to the formal proof community, in the way it deliberates
carefully on matters of notation and terminology, finds the appropriate
level of generalization of concepts, and situates
different branches of mathematics within a coherent framework.

The opening quote claims that formalized mathematics is absolutely
unrealizable.  Bourbaki objected that formal proofs are too long
(``{\it la moindre d\'emonstration \ldots exigerait d\'ej\`a des
  centaines de signes}''), that it would be a burden to forego the
convenience of abuses of notation, and that they do not leave room for
useful metamathematical arguments and
abbreviations~\cite{bourbaki1966theorie}.

Bourbaki is correct in the strict sense that no human artifact is
absolutely trustworthy and that the standards of mathematics evolve in
a historical process, according to available technology.
Nevertheless, the technological barriers hindering formalization have
fallen one after another.  Today, computer verifications that check
millions of inferences are routine.  As Gonthier has convincingly
shown in the Odd Order theorem project, many abuses of notation can
actually be described by precise rules and implemented as algorithms,
making the term {\it abuse of notation} really something of a
misnomer, and allowing mathematicians to work formally with
customary notation.  Finally, the trend over the past decades has been
to move more and more features out of the metatheory and into the
theory by making use of features of higher-order logic and reflection.  In
particular, it is now standard to treat abbreviations and definitions
as part of the logic itself rather than metatheory.

\subsection{Future work}

This report has described three major projects in the world of formal
proofs: trustworthy systems with HOL, advanced mathematical theorems
formalized in Coq, and increased automation.

We are still far from having an automated mathematical journal referee
system, but close enough to propose this as a realistic research
program.  Already some 10\% of all papers of the Principles of
Programming Languages (POPL) symposium in computer science are
completely formalized \cite{SewPOPL2014}.  Other recent research automates the
translation of mathematical prose from English into a
computer-parsable form with semantic content
\cite{ganesalingam2013language}.  As these technologies develop, we
may anticipate the day when the precise formal statements of
mathematical theorems may be extracted from the prose.  Once
sufficiently many statements from the natural language proof can
similarly be extracted, proof automation will take over, filling in
the remaining details, to produce a formal proof from the natural
language text.

For other surveys of formal proofs, see \cite{avigad2014formally},
\cite{Hales:2008:formal}.

\section{Appendix. Some formally verified theorems}

This appendix gives examples of some theorems that have
been successfully formalized in various proof assistants.  The purpose
of these examples is to showcase the range of what can be obtained by
current technologies.

The Four-Color theorem was formalized in Coq~\cite{gonthier2008formal}.

\begin{lstlisting}[keepspaces=true,stringstyle=\tt,basicstyle=\small,frame=single,framesep=8pt,morekeywords={Variable,Theorem,Proof,Qed},columns=flexible]
Variable R : real_model. 
Theorem four_color : (m : (map R))
     (simple_map m) -> (map_colorable (4) m). 
Proof.
    Exact (compactness_extension four_color_finite). 
Qed.
\end{lstlisting}

The elementary proof of the Prime Number Theorem by Erd\"os and
Selberg was formalized in Isabelle \cite{avigad2007formally}.  The
analytic proof by Hadamard and de la Vall\'ee Poussin was formalized
in HOL Light~\cite{harrison2009formalizing}. In the statement that
follows, the symbol \verb!&! denotes the function embedding the
natural numbers into the real numbers.

\begin{lstlisting}[keepspaces=true,stringstyle=\tt,basicstyle=\small,frame=single,framesep=8pt,mathescape,morekeywords={},columns=flexible]
  // Prime Number Theorem:
  ((\n. &(CARD {p | prime p $\land$ p <= n}) / (&n / log(&n)))
    ---> &1) sequentially
\end{lstlisting}

Here is the formal statement of the Brouwer fixed point theorem, which was formalized in HOL Light by Harrison.

\begin{lstlisting}[keepspaces=true,stringstyle=\tt,basicstyle=\small,frame=single,framesep=8pt,mathescape,morekeywords={},columns=flexible]
  $\forall$ f:real$^N$->real$^N$ s. 
  compact s $\land$ convex s $\land$ ~(s = {}) $\land$ f continuous_on s $\land$ IMAGE f s SUBSET s
  $\Longrightarrow$ $\exists$x. x IN s $\land$ f x = x
\end{lstlisting}

The formalization of the central limit theorem was carried out earlier this year in Isabelle \cite{avigad2014central}.

\begin{lstlisting}[keepspaces=true,stringstyle=\tt,basicstyle=\small,frame=single,framesep=8pt,mathescape,morekeywords={theorem,fixes,assumes,defines,shows,Variable,Theorem,Proof,Qed},columns=flexible]
theorem (in prob_space) central_limit_theorem:
  fixes 
    X :: "nat $\Rightarrow$ 'a $\Rightarrow$ real" and
    $\mu$ :: "real measure" and
    $\sigma$ :: real and
    S :: "nat $\Rightarrow$ 'a $\Rightarrow$ real"
  assumes
    X_indep: "indep_vars ($\lambda$i. borel) X UNIV" and
    X_integrable: "$\bigwedge$n. integrable M (X n)" and
    X_mean_0: "$\bigwedge$n. expectation (X n) = 0" and
    $\sigma$_pos: "$\sigma$ > 0" and
    X_square_integrable: "$\bigwedge$n. integrable M ($\lambda$x. (X n x)$^2$)" and
    X_variance: "$\bigwedge$n. variance (X n) = $\sigma^2$" and
    X_distrib: "$\bigwedge$n. distr M borel (X n) = $\mu$"
  defines
    "S n $\equiv$ $\lambda$x. $\Sigma\,$i<n. X i x"
  shows
    "weak_conv_m ($\lambda$n. distr M borel ($\lambda$x. S n x / sqrt (n * $\sigma^2$))) 
        (density lborel standard_normal_density)"
\end{lstlisting}


The Kepler conjecture asserts that no packing of congruent balls in
$\ring{R}^3$ can have density greater than the face-centered cubic
packing.  The Kepler conjecture is a theorem whose proof relies on
many computer calculations~\cite{Hales:2006:DCG}.  The 
Kepler conjecture has been formalized\footnote{This project was completed on August 10, 2014.} in a combination of the HOL
Light and Isabelle proof assistants~\cite{website:FlyspeckProject}.  This formalization has
been a large collaborative effort. 

\begin{lstlisting}[keepspaces=true,stringstyle=\tt,basicstyle=\small,frame=single,framesep=8pt,framextopmargin=10pt,mathescape,morekeywords={Variable,Theorem,Proof,Qed},columns=flexible]
 $\vdash$ the_nonlinear_inequalities

 $\vdash$ import_tame_classification $\land$      
    the_nonlinear_inequalities $\land$
    $\Longrightarrow$ the_kepler_conjecture

 $\vdash$ the_kepler_conjecture $\Longleftrightarrow$
     ($\forall$V. packing V
            $\Longrightarrow$ ($\exists$c. $\forall$r. &1 $\le$ r
                         $\Longrightarrow$ &(CARD(V $\cap$ ball(vec 0,r))) $\le$
                             $\pi$ * r$^3$ / sqrt(&18) + c * r$^2$))
\end{lstlisting}

\section{Appendix. The inference rules of HOL}

The type system, the terms, sequents, and axioms of HOL have been
described in the first section of this report.  For reference purposes, we list
all the inference rules of HOL, as formulated by Harrison.  We borrow from
the presentation in \cite{Hales:2008:formal}.

The system has ten inference rules and a mechanism for defining new
constants and types. Each inference rule is depicted as a fraction;
the inputs to the rule are listed in the numerator, and the output in
the denominator.  The inputs to the rules may be terms or other
theorems.  In the following rules, we assume that $p$ and $p'$ are
equal, up to a renaming of bound variables, and similarly for $b$ and
$b'$.  (Such terms are called $\alpha$-equivalent.)

\quad On first reading, ignore the assumption lists $\Gamma$ and
$\Delta$. They propagate silently through the inference rules, but are
really not what the rules are about.  When taking the union
$\Gamma\cup\Delta$, $\alpha$-equivalent assumptions should be
considered as equal.  \smallskip


\smallskip

\noindent
Equality is reflexive:
$$
\frac{a}{\vdash a=a}
$$

\noindent
Equality is transitive:
$$
\frac{\Gamma \vdash a=b;~~~\Delta\vdash b'=c}
{\Gamma\cup\Delta \vdash a=c}
$$

\noindent
Equal functions applied to equals are equal:
$$
\frac{\Gamma\vdash f=g;~~~\Delta\vdash a=b}
{\Gamma\cup\Delta\vdash f\hskip0.1em a = g\hskip0.1em b}
$$

\noindent
The rule of abstraction holds. Equal function bodies
give equal functions:
$$
\frac{x;~~~\Gamma\vdash a=b}
{\Gamma \vdash \lambda x.\ a~=\lambda x.\ b}
~\hbox{\ (if $x$ is not free in $\Gamma$)}
$$

\noindent
The application of the function $x\mapsto a$ to $x$ gives $a$:
$$
\frac{(\lambda x.~a)\, x}
{\vdash (\lambda x.\ a)\, x = a}
$$




\noindent
Assume $p$, then conclude $p$:
$$
\frac{p\tc bool}
{p \vdash p}
$$

\noindent
An `equality-based' rule of modus ponens holds:
$$
\frac{\Gamma\vdash p;~~~\Delta \vdash p'=q}
{\Gamma\cup \Delta \vdash q}
$$

\noindent
If the assumption $q$ gives conclusion $p$ and the assumption $p$
gives $q$, then they are equivalent:
$$
\frac{\Gamma \vdash p;~~~\Delta\vdash q}
{(\Gamma\setminus q)\cup (\Delta\setminus p)
\vdash p=q}
$$

Type variable substitution holds.  If arbitrary types are substituted
in parallel for type variables in a sequent, a theorem results.  Term
variable substitution holds.  If arbitrary terms are substituted in
parallel for term variables in a sequent, a theorem results.



{\it Acknowledgements.}   I would like to thank the many people who helped me
during the preparation of this report, particularly Jeremy Avigad, John Harrison, Chris Kapulkin, 
W\"oden Kusner,  and Josef Urban.  I wish to give special thanks to Assia Mahboubi for answering many questions related
to Coq and the Odd Order theorem.
I would also like to thank the many speakers and participants at the special program on {\it Semantics of proofs and certified 
math} at IHP.

\raggedright
\bibliographystyle{plain} 
\bibliography{/Users/flyspeck/Desktop/googlecode/flyspeck/usr/thales/latex/bibliography/all}

\end{document}